\documentclass[a4paper]{article}

\usepackage{amsmath}
\usepackage{amsthm}
\usepackage{amssymb}
\usepackage{fullpage}
\usepackage[utf8]{inputenc}
\usepackage{microtype}
\usepackage[hidelinks]{hyperref}

\newcommand{\SEP}{{\diamond}}

\newtheorem{thm}{Theorem}[section]
\newtheorem{theorem}[thm]{Theorem}

\newtheorem{lemma}[thm]{Lemma}

\theoremstyle{definition}
\newtheorem{mydefinition}[thm]{Definition}
\newtheorem*{claim}{Claim}
\theoremstyle{remark}
\newtheorem{myexample}[thm]{Example}

\usepackage{authblk}

\newcommand{\cqed}{}

\begin{document}
\title{On Two LZ78-style Grammars: Compression Bounds and Compressed-Space Computation
}
\author[1]{Golnaz Badkobeh\thanks{Supported by the Leverhulme Trust on the Leverhulme Early Career Scheme.}}
\author[2]{Travis Gagie}
\author[3]{Shunsuke Inenaga}
\author[4]{Tomasz Kociumaka\thanks{Supported by Polish budget funds for science in 2013--2017 under the `Diamond Grant' program.}}
\author[5]{Dmitry~Kosolobov}
\author[5]{Simon J. Puglisi\thanks{Supported by the Academy of Finland via grant 294143.}}

\date{\vspace{-1cm}}

\affil[1]{Department of Computer Science, University of Warwick, Coventry, England}
\affil[ ]{\href{mailto:g.badkobeh@warwick.ac.uk}{\url{g.badkobeh@warwick.ac.uk}}}
\affil[2]{EIT, Diego Portales University and CeBiB, Santiago, Chile}
\affil[ ]{\href{mailto:travis.gagie@mail.udp.cl}{\url{travis.gagie@mail.udp.cl}}}
\affil[3]{Department of Informatics, Kyushu University, Fukuoka, Japan}
\affil[ ]{\href{mailto:inenaga@inf.kyushu-u.ac.jp}{\url{inenaga@inf.kyushu-u.ac.jp}}}
\affil[4]{Institute of Informatics, University~of~Warsaw,~Warsaw,~Poland}
\affil[ ]{\href{mailto:kociumaka@mimuw.edu.pl}{\url{kociumaka@mimuw.edu.pl}}}
\affil[5]{Department of Computer Science, University of Helsinki, Helsinki, Finland}
\affil[ ]{\href{mailto:dkosolobov@mail.ru}{\url{dkosolobov@mail.ru}} \href{mailto:puglisi@cs.helsinki.fi}{\url{puglisi@cs.helsinki.fi}}}

\maketitle

\begin{abstract}
We investigate two closely related LZ78-based compression schemes: LZMW (an old scheme by Miller and Wegman) and LZD (a recent variant by Goto et al.). Both LZD and LZMW naturally produce a grammar for a string of length $n$; we show that the size of this grammar can be larger than the size of the smallest grammar by a factor $\Omega(n^{\frac{1}3})$ but is always within a factor $O((\frac{n}{\log n})^{\frac{2}{3}})$. In addition, we show that the standard algorithms using $\Theta(z)$ working space to construct the LZD and LZMW parsings, where $z$ is the size of the parsing, work in $\Omega(n^{\frac{5}4})$ time in the worst case. We then describe a new Las Vegas LZD/LZMW parsing algorithm that uses $O (z \log n)$ space and $O(n + z \log^2 n)$ time with high probability.
\vspace{1ex}

\end{abstract}

\section{Introduction}

The LZ78 parsing~\cite{ZL78} is a classic dictionary compression technique, discovered by Lempel and Ziv in 1978, that gained wide use during the 1990s in, for example, the Unix {\tt compress} tool and the GIF image format. Not written about until much later was that LZ78 actually produces a representation of the input string as a context-free grammar. In recent years, grammar compressors have garnered immense interest, particularly in the context of compressed text indexing: it is now possible to efficiently execute many operations directly on grammar-compressed strings, without resorting to full decompression (e.g., see \cite{BCPT15,BLRSSW15,CN11,GGKNP12,INIBT13,TIIBT13}).

A wide variety of grammar compressors are now known, many of them analyzed by Charikar et al.~\cite{CharikarEtAl} in their study of the smallest grammar problem, which is to compute the smallest context-free grammar that generates the input string (and only this string). Charikar et al.~show that this problem is NP-hard, and further provide lower bounds on approximation ratios for many grammar compressors. LZ78 is shown to approximate the smallest grammar particularly poorly, and can be larger than the smallest grammar by a factor $\Omega(n^\frac{2}{3}/\log n)$ (in~\cite{HuckeLohreyReh} this bound was improved to $\Omega((\frac{n}{\log n})^{\frac{2}{3}})$), where $n$ is the input length.

Our focus in this paper is on the LZD~\cite{GotoEtAl} and LZMW~\cite{MillerWegman} grammar compression algorithms, two variants of LZ78 that usually outperform LZ78 in practice. Despite their accepted empirical advantage over LZ78, no formal analysis of the compression performance of LZD and LZMW in terms of the size of the smallest grammar exists. This paper addresses that need. Moreover, we show that the standard algorithms for computing LZD and LZMW have undesirable worst case performance, and provide an alternative algorithm that runs in log-linear randomized time. In particular the contributions of this article are as follows:
\begin{enumerate}
\item We show that the size of the grammar produced by LZD and LZMW can be larger than the size of the smallest grammar by a factor $\Omega(n^{\frac{1}{3}})$ but is always within a factor $O((\frac{n}{\log n})^{\frac{2}{3}})$. To our knowledge these are the first non-trivial bounds on compression performance known for these algorithms.
\item Space usage during compression is often a concern. For both LZD and LZMW, parsing algorithms are known that use $O(z)$ space, where $z$ is the size of the final parsing. We describe strings for which these algorithms require $\Omega(n^{\frac{5}{4}})$ time. (The only previous analysis is an $O(n^2/\log n)$ upper bound~\cite{GotoEtAl}.)
\item We describe a Monte-Carlo parsing algorithm for LZD/LZMW that uses a z-fast trie~\cite{BBV10} and an AVL-grammar~\cite{Rytter03} to achieve $O(z \log n)$ space and $O(n + z \log^2 n)$ time for inputs over the integer alphabet $\{0,1,\ldots,n^{O(1)}\}$. This algorithm works in the streaming model and computes the parsing with high probability. Using the Monte-Carlo solution, we obtain a Las Vegas algorithm that, with high probability, works in the same space and time.
\end{enumerate}

In what follows we provide formal definitions and examples of LZD and LZMW parsings. Section~\ref{sec-approx} then establishes bounds for the approximation ratios for the sizes of the LZD/LZMW grammars. In Section~\ref{sec-smallspace} we consider the time efficiency of current space-efficient parsing schemes for LZD/LZMW. Section~\ref{sec-faster} provides an algorithm with significantly better (albeit randomized) performance. Conclusions and reflections are offered in Section~\ref{sec-conclusion}.

\subsection{Preliminaries.}
We consider strings drawn from an alphabet $\Sigma$ of size $\sigma = |\Sigma|$.
The \emph{empty string} is denoted by $\epsilon$.
The $i$th letter of a string $s$ is denoted by $s[i]$ for $i$ such that $1 \le i \le |s|$, and the substring of $s$ that begins at position $i$ and ends at position $j$ is denoted by $s[i..j]$ for $1 \le i \le j \le |s|$.
Let $s[i..j] = \epsilon$ if $j<i$.
For any $i,j$, the set $\{k\in \mathbb{Z} \colon i \le k \le j\}$ (possibly empty) is denoted by $[i..j]$.

For convenience, we assume that the last letter of the input string $s$ is  $ \$$, where $ \$$ is a special delimiter letter that does not occur elsewhere in the string.

\begin{mydefinition}The \emph{LZD (LZ--Double) parsing} \cite{GotoEtAl} of a string $s$ of
length $n$ is the parsing $s = p_1 p_2 \cdots p_z$ such that, for $i \in [1..z]$, $p_i= p_{i_1}p_{i_2}$ where $p_{i_1}$ is the longest prefix of $s[k..n]$ and $p_{i_2}$ is the longest prefix of $s[k+|p_{i_1}|..n]$ with $p_{i_1}, p_{i_2}\in \{p_1,\ldots,p_{i-1}\}\cup \Sigma $ where $k=|p_1\cdots p_{i-1}|+1$. We refer to the set $\Sigma\cup\bigcup_{i \in [1..z]}\{p_i\}$ as the \emph{dictionary of LZD}.
\end{mydefinition}

\begin{mydefinition}
The \emph{LZMW (LZ--Miller--Wegman) parsing} \cite{MillerWegman} of a string $s$ of
length $n$ is the parsing $s = p_1 p_2 \cdots p_z$ such that,
for $i \in [1..z]$, $p_i$ is the longest prefix of $s[k..n]$  with $p_i \in \{p_j p_{j+1}\colon 1 \le j \le i-2\}\cup \Sigma$ where $k=|p_1\cdots p_{i-1}|+1$. We refer to the set $\bigcup_{i \in [2..z]}\{p_{i-1}p_i\}$ as the \emph{dictionary of LZMW}.
\end{mydefinition}

\begin{myexample}
The LZD parsing of the string $s = abbaababaaba\$$ is $p_1=ab$, $p_2=ba$, $p_3=abab$, $p_4=aab$, and $p_5=a\$$. This can be represented by $(a,b), (b,a), (1, 1),(a, 1), (a,\$)$.
The LZMW parsing of $s$ is the following: $p_1=a$, $p_2=b$, $p_3=b$, $p_4=a$, $p_5=ab$, $p_6=ab$, $p_7=aab$, $p_8= a$, and $p_{9}=\$$. This can be represented by $(a,b,b,a,1,1,4,a,\$)$.
\end{myexample}

Notice that the LZD/LZMW parsing of string $s$ can be seen as a grammar that only generates $s$, with production rules of form $p_i \rightarrow p_jp_k~(j < i, k < i)$ or $p_i \rightarrow a~(\in \Sigma)$ for each phrase $p_i$, and the start rule $S \rightarrow p_1 p_2 \cdots p_z$. The \emph{size} of a grammar is the total number of symbols in the right-hand side of the production rules. Thus, the size of the LZD (resp., LZMW) grammar is only by a constant factor larger than the number of phrases in the LZD (resp., LZMW) parsing.

\section{Approximating the Smallest Grammar}
\label{sec-approx}
The following theorem shows that, although LZD and LZMW have good compression performance in practice on high-entropy strings, their performance on low-entropy strings can be very poor.

\begin{theorem}
For arbitrarily large $n$, there are strings $s$ of length $n$ for which the size of the grammars produced by the LZD and LZMW parsings is larger than the size of the smallest grammar generating $s$ by a factor $\Omega(n^{\frac{1}3})$.
\end{theorem}
\begin{proof}
Our proof is inspired by~\cite[Section VI, C]{CharikarEtAl}. Let $k \ge 4$ be an integer that is a power of~$2$. We will construct a string $s$ of length $n = \Theta(k^3)$ that can be encoded by a grammar of size $O(k) = O(n^{\frac{1}3})$, but for which the LZMW parsing produces a grammar of size $\Omega(k^2) = \Omega(n^{\frac{2}3})$. The input alphabet is $\{a,b,c,d\}$; the letters $c$ and $d$ serve as separators. Denote $\delta_i = a^i bb a^{k-i}$ and $\gamma_i = ba^i\,a\,a^ib\,c\,ba\,ba^2\,ba^3\cdots ba^i$. The string $s$ is as follows:
$$
\begin{array}{l}
x = \delta_k \delta_{k-1}\, \delta_k \delta_{k-2}\, \delta_k \delta_{k-3} \cdots \delta_k \delta_{k/2+1}\, \delta_k a^{k-1},\\
s = \gamma_0\gamma_1\cdots\gamma_{k-1}\delta_0 d\delta_1 d\cdots \delta_k d\,c aa\,c aa^2a^2 \cdots ca^{2^i-1}a^{2^i}a^{2^i}\cdots ca^{\frac{k}2-1}a^{\frac{k}2}a^{\frac{k}2} dc\, x^{\frac{k}2}\enspace.
\end{array}
$$

We have $|s| = \Theta(k^3)$. Consider the prefix $\gamma_0\gamma_1 \cdots \gamma_{k-1}$ $\delta_0 d\delta_1 d \cdots d\delta_k d$, which will ensure the strings $\delta_i$ are in the LZMW dictionary.

We will show by induction on $i$ that each substring $\gamma_i$ of the prefix $\gamma_0\gamma_1\cdots \gamma_{k-1}$ is composed of the phrases $ba^i$, $a$, $a^ib$, $cbaba^2\cdots ba^i$ in the parsing of the string $s$. It is trivial for $i = 0$. Suppose that $i > 0$ and the assertion holds for all $\gamma_{i'}$ and $i' < i$. It follows from the inductive hypothesis that $ba^i$ is the longest prefix of $\gamma_i$ that is equal to a concatenation of two adjacent phrases introduced before the starting position of $\gamma_i$. Hence, by the definition of LZMW, the string $\gamma_i$ starts with the phrase $ba^i$. In the same way we deduce that the phrase $ba^i$ is followed by the phrases $a$, $a^ib$, and $cbaba^2\cdots ba^i$.

By a similar inductive argument, one can show that each substring $\delta_i d$ of the substring $\delta_0 d\delta_1 d\cdots \delta_k dc$ is composed of the phrases $a^ib$, $ba^{k-i}$, $d$. Since the phrases $a^ib$ and $ba^{k-i}$ are adjacent, the LZMW dictionary now contains the strings $\delta_i = a^ibba^{k-i}$ for all $i = 0,1,\ldots, k$.

Similarly, the substring $c aa caa^2a^2 \cdots ca^{2^i-1}a^{2^i}a^{2^i} \cdots ca^{\frac{k}2-1}a^{\frac{k}2}a^{\frac{k}2} dc$ is parsed as $c, a, a, ca, a^2, a^2, \ldots,$ $ca^{2^i-1}, a^{2^i}, a^{2^i}, \ldots, ca^{\frac{k}2-1}, a^{\frac{k}2}, a^{\frac{k}2}, dc$. In what follows we need only the string $a^k$ introduced to the dictionary by the pair of phrases $a^{\frac{k}2}$.

Finally, consider the substring $x^{\frac{k}2}$. Observe that the first occurrence of $x$ is parsed in (almost) the way it is written, i.e., it is parsed as $\delta_k, \delta_{k-1}, \delta_k, \delta_{k-2}, \ldots, \delta_k, \delta_{k/2+1}, \delta_k$. But the last phrase is $a^k$ instead of $a^{k-1}$. In other words, the parsing of the second occurrence of $x$ starts from the second position of $x$ and, therefore, the first phrases of this parsing are as follows:
$$
\delta_{k-1}, \delta_{k-2}, \delta_{k-1}, \delta_{k-3}, \ldots, \delta_{k-1}, \delta_{k/2}, \delta_{k-1}.
$$
Again, the last phrase is $a^k$ and, hence, the parsing of the third occurrence of $x$ starts with the third position of $x$, and so on.

The LZMW parsing of $s$, therefore, consists of $\Omega(k^2)$ phrases and the size of the LZMW grammar is $\Omega(k^2)$. But there is a grammar of size $O(k)$ producing $s$:
$$
\begin{array}{l}
S \rightarrow \Gamma_0\Gamma_1\cdots\Gamma_{k-1}\Delta_0d\Delta_1d\cdots \Delta_kd cA_{2}cA_{5}cA_{11}\cdots cA_{k/2+k-1}dcX^{k/2},\\
A_0 \rightarrow \epsilon, \quad B_0 \rightarrow c,\quad A_i \rightarrow A_{i-1}a, \quad B_i \rightarrow B_{i-1}bA_i\quad\text{ for }i \in [1..2k],\\
\Gamma_i \rightarrow bA_{2i+1}bB_i, \quad \Delta_i \rightarrow A_i bb A_{k-i}\quad \text{ for }i \in [0..k],\\
X \rightarrow \Delta_k\Delta_{k-1}\,\Delta_k\Delta_{k-2}\cdots \Delta_k\Delta_{k/2+1}\,\Delta_k A_{k-1}\enspace.
\end{array}
$$

Using similar ideas we can describe a troublesome string for the LZD scheme:
$$
s = (a^2\,c^2\,a^3\,c^3 \cdots a^{k}c^{k})(bb\,abb\,a^2bb\,a^3 \cdots bba^{k-1}bb)(\delta_0d^2\delta_1d^3 \cdots \delta_kd^{k+2})x^{\frac{k}2}\enspace.
$$

As above, the size of the grammar corresponding to the LZD parsing of $s$ is $\Omega(k^2)$ whereas the size of the smallest grammar is $O(k)$; hence, the result follows.
$$
\begin{array}{l}
S \rightarrow A_2C_2A_3C_3\cdots A_kC_kbbA_1bbA_2\cdots bbA_{k-1}bb\Delta_0D_2\Delta_1D_3\cdots \Delta_kD_{k+2}X^{k/2},\\
A_0 \rightarrow \epsilon, C_0 \rightarrow \epsilon, D_0 \rightarrow \epsilon, A_i \rightarrow A_{i-1}a, C_i \rightarrow C_{i-1}c, D_i \rightarrow D_{i-1}d\text{ for }i \in [1..k{+}2],\\
\Delta_i \rightarrow A_i bb A_{k-i} \text{ for }i \in [0..k],\quad
X \rightarrow \Delta_k\Delta_{k-1}\,\Delta_k\Delta_{k-2}\cdots \Delta_k\Delta_{k/2+1}\,\Delta_k A_{k-1}\enspace.
\end{array}
$$
The analysis is similar to the above but simpler, so, we omit it. To additionally verify the correctness of both constructions, we conducted experiments on small $k$ and, indeed, observed the described behavior; the code can be found in~\cite{LZDandLZMWcode}.
\cqed
\end{proof}

We can also show that the upper bound for the approximation ratio of the LZ78 parsing given in~\cite{CharikarEtAl} also applies to the LZD and LZMW parsings. For this, we will use the following known results.

\begin{lemma}[\cite{CharikarEtAl}] \label{lem:mkLemma}
If there is a grammar of size $m$ generating a given string,
then this string contains at most $mk$ distinct substrings of length $k$.
\end{lemma}

\begin{lemma}[\cite{GotoEtAl}] \label{lem:LZD-distinct}
All phrases in the LZD parsing of a given string are distinct.
\end{lemma}

\begin{lemma}\label{lem:LZMW-distinct}
Let $p_1p_2\cdots p_z$ be the LZMW parsing of a given string. Then, for any $i \in [2..z]$ and $j \in [i{+}2 .. z]$, we have $p_{i-1}p_i \ne p_{j-1}p_j$.
\end{lemma}
\begin{proof}
If $p_{i-1}p_i = p_{j-1}p_j$ for $i < j - 1$, then, by the definition of LZMW, the phrase $p_{j-1}$ either is equal to $p_{i-1}p_i$ or contains $p_{i-1}p_i$ as a prefix, which is a contradiction.
\cqed
\end{proof}

Now we are ready to show an upper bound on the approximation ratio of the LZD and LZMW parsings.
\begin{theorem} \label{theo:upperbound-LZD}
For all strings $s$ of length $n$, the size of the grammar produced by the LZD/LZMW parsing is larger than the size of the smallest grammar generating $s$ by at most a factor $O((n/\log n)^{2/3})$.
\end{theorem}
\begin{proof}
The theorem can be shown by an analogous way as for the upper bound of the LZ78 parsing against the smallest grammar~\cite{CharikarEtAl} (which is especially straightforward for LZD due to Lemma~\ref{lem:LZD-distinct}), but we provide a full proof for completeness.

Let us consider LZMW. Suppose that $s$ is a string of length $n$ and $m^*$ is the size of the smallest grammar generating $s$. Let $p_1, p_2, \ldots, p_z$ be the LZMW parsing of $s$. It suffices to evaluate the number $z$ of phrases since the total size of the grammar produced by LZMW is only by a constant factor larger than $z$.

Consider the multiset $S = \{p_1p_2, p_2p_3, \ldots, p_{z-1}p_z\}$ (recall that a multiset can contain an element more than one time). Let $p_{i_1}p_{i_1+1}, p_{i_2}p_{i_2+1}, \ldots, p_{i_{z-1}}p_{i_{z-1}+1}$ be a sequence of all strings from $S$ sorted in increasing order of their lengths (again, some strings may occur more than once in the sequence). We partition the sequence by grouping the first $2\cdot m^*$ strings, then the next $2\cdot 2m^*$ strings, the next $2\cdot 3m^*$ strings, and so forth. Let $r$ be the minimal integer satisfying $2(1m^* + 2m^* + \cdots + rm^* + (r+1)m^*) > z$. This implies that $z = O(r^2m^*)$.

By Lemma~\ref{lem:LZMW-distinct}, any string has at most two occurrences in the multiset $S$. Also, it follows from Lemma~\ref{lem:mkLemma} that $s$ contains at most $km^*$ distinct substrings of length $k$. Thus, for any $k \ge 1$, there are at most $2km^*$ strings from $S$ that generate substrings of length $k$. This implies that each string in the $k$th group generates a substring of length at least $k$. Hence, we have that
\[
  2n \ge |p_{i_1}p_{i_1+1}| + |p_{i_2}p_{i_2+1}| + \cdots + |p_{i_{z-1}}p_{i_{z-1}+1}| \ge 2(1^2m^* + 2^2m^* + \cdots + r^2m^*),
\]
which implies that $r = O((n/m^*)^{1/3})$.
By plugging this into $z = O(r^2m^*)$, we obtain $z = O((n/m^*)^{2/3}m^*)$ and thus the approximation ratio of the grammar produced by LZMW is $O((n / m^*)^{2/3})$. Since $m^* = \Omega(\log n)$, we finally get the desired bound $O((n/\log n)^{2/3})$.

Let us sketch the analysis of LZD, which is very similar. In this case, we consider the set $S'$ of all phrases $p_1, p_2, \ldots, p_z$ (not pairs as in LZMW) of the LZD parsing. Let $p_{i_1}, \ldots, p_{i_z}$ be the sequence of all strings from $S'$ sorted by the increasing order of lengths. We partition the sequence into groups of size $1m^*, 2m^*, 3m^*, \ldots$ (without the factor~$2$ as in LZMW). It follows from Lemma~\ref{lem:LZD-distinct} that any string occurs in $S'$ at most once. Therefore, similar to the case of LZMW, we obtain $n = |p_{i_1}| + |p_{i_2}| + \cdots + |p_{i_z}| \ge 1^2m^* + 2^2m^* + \cdots + r^2m^*$, which implies the result in the same way as above.
\cqed
\end{proof}

\section{Small-Space Computation}
\label{sec-smallspace}

In this section we analyze the time required to compute the LZD and LZMW parsings
using the $O(z)$-space algorithms described by Goto et al.~\cite{GotoEtAl} and
Miller and Wegman~\cite{MillerWegman}, where $z$ is the number of phrases.
We focus on LZD throughout, but a very similar algorithm and analysis applies for LZMW.
Goto et al. upperbound the runtime at $O(z(m + \min(z,m)\log\sigma))$, where $m$
is the length of the longest LZD (or LZMW) phrase and $\sigma$ is the size of the input alphabet.
Because $m = O(n)$ and $z = O(n)$,
the runtime is upper bounded by $O(n^2)$. Below we provide a lower bound of $\Omega(n^{5/4})$
on the worst-case runtime, but before doing so we provide the reader with a description of Goto~et~al.'s
algorithm~\cite{GotoEtAl}.\footnote{We concern ourselves here with LZD parsing, but it should be easy for the reader to
see that the algorithms are trivially adapted to instead compute LZMW.}

\paragraph{Na\"{\i}ve parsing algorithms.}

In the compacted trie for a set of strings, each edge label $\ell$ is represented
as a pair of positions delimiting an occurrence of $\ell$ in the set.
In this way we can store the trie for \(s_1, \ldots, s_k\) in $O(k)$ space.
During parsing Goto et al.~\cite{GotoEtAl} maintain the dictionary of LZD phrases in a compacted
trie.
The trie is of size $O(z)$, but read-only random access
to the input string is also required in order to determine the actual values of the
strings on the edge labels.

Initially the trie is empty, consisting of only the root. At a generic step during parsing,
when we go to compute the phrase $p_i = p_{i_1}p_{i_2}$ starting at position
$j = |p_1p_2\ldots p_{i-1}| + 1$, the trie contains nodes representing the phrases
$p_1, p_2, \ldots, p_{i-1}$ and all the distinct symbols occurring in $s[1..j-1]$, and all these
nodes (corresponding to phrases and symbols) are marked. Note that there may also be some nodes in the trie that
do not correspond to any phrase, i.e., branching nodes. Let $s[j..k]$ be the longest prefix of $s[j..n]$
that can be found by traversing the trie from the root. If $s[j..k]$ cannot be matched even for $k = j$,
then $s[j]$ is the leftmost occurrence of symbol $c = s[j]$ in $s$, and we add a child node of the
root labelled with $c$, mark the node, and set it as the first element of the new phrase, i.e., $p_{i_1} = c$.
Otherwise, the first element of $p_i$, $p_{i_1}$, is the string written on the path connecting the root and
the lowest marked node on the path that spells $s[j..k]$.
The second element, $p_{i_2}$, of the phrase is computed in a
similar manner, by searching for $s[j+|p_{i_1}|+1..n]$ in the trie.

After computing $p_i$ we modify the trie by a standard procedure so that there is a marked node representing $p_i$:
first, we traverse the trie from the root finding the longest prefix of $p_i$ present in the trie,
then, possibly, create one or two new nodes, and, finally, mark the node (which, probably, did not exist before)
corresponding to $p_i$ (the details can be found in any stringology textbook).

The time taken to compute a new phrase and update the trie afterwards is bounded by
$O(m + \min(z,m)\log\sigma)$, where $m = O(n)$ is the length of the longest phrase
(and therefore an upper bound on the length of the longest
path in the trie), $\min(z,m)$ is an upper bound on the number of branching nodes, and $\log\sigma$ is
the time taken to find the appropriate outgoing edge at each branching node during downward traversal.
Over all $z$ phrases the runtime is thus $O(z(m + \min(z,m)\log\sigma))$.

The LZMW construction algorithm of Miller and Wegman~\cite{MillerWegman} is analogous but, unlike the LZD algorithm, when we go to compute
the phrase $p_i$, the trie contains the strings $p_1p_2, p_2p_3, \ldots, p_{i-2}p_{i-1}$ and the nodes corresponding to these strings are marked. One can easily show that the running time of this algorithm is $O(z(m + \min(z,m)\log\sigma))$, where $z$ and $m$ are defined analogously as for LZD.

We call both these algorithms \emph{na\"{\i}ve}.

\paragraph{Worst-case time of the na\"{\i}ve algorithms.}

Now let us investigate the worst-case time complexity of the na\"{\i}ve LZD and LZMW construction algorithms.
\begin{theorem}
The na\"{\i}ve LZD and LZMW construction algorithms take time $\Omega(n^{\frac{5}4})$ in the worst case.\label{LowerboundThm}
\end{theorem}
\begin{proof}
Let $k \ge 8$ be an integer that is a power of two. We will describe a string $s$ of length $n = \Theta(k^4)$ for which the basic LZD construction algorithm (see the above discussion) spends $\Theta(n^{\frac{5}4})$ time to process. The string $s$ is composed of pairwise distinct letters $a_{i,j}$, for $i, j \in [1..k]$, and ``separator'' letters, all of which are denoted $\SEP$ and supposed to be distinct. We will first construct a prefix $s'$ of $s$ that forces the algorithm to fill the dictionary with a set of strings that are used as building blocks in further constructions. To this end, denote (with parentheses used only for convenience):
$$
\begin{array}{l}
w_i = a_{i,1}a_{i,2} \cdots a_{i,k}\text{ for }i = 1,2,\ldots,k\text{ and }w = w_1w_2 \cdots w_k,\\
s_{pre,i} = w_i[1..2]w_i[1..3] \cdots w_i[1..k]\text{ for }i = 1,2,\ldots,k,\\
s_{suf,i} = w_i[k{-}1..k]w_i[k{-}2..k] \cdots w_i[2..k]\text{ for }i = 1,2,\ldots,k,\\
p = (s_{pre,1}s_{pre,2} \cdots s_{pre,k}) (s_{suf,1}s_{suf,2} \cdots s_{suf,k}),\\
q = (w_{k-2}w_{k-1})(w_{k-3}w_{k-2}w_{k-1}) \cdots (w_1w_2{\cdots}w_{k-1}) (w),\\
s' = pq\cdot w^{2^1} w^{2^2} \cdots w^k (w_k[2..k]w^k) (w_k[3..k]w^k) \cdots (w_k[k..k]w^k).
\end{array}
$$
Analyzing the prefix $p$ of $s'$, it is clear that the LZD construction algorithm adds to the dictionary exactly all prefixes and suffixes of the strings $w_i$ for $i = 1,2,\ldots, k$; parsing the string $q$, the algorithm adds the strings $w_{k-2}w_{k-1}, w_{k-3}w_{k-2}w_{k-1}, \ldots, w_1w_2\cdots w_{k-1}$, and $w_1w_2 \cdots w_k = w$; then, processing the string $w^{2^1} w^{2^2} \cdots w^k$, the algorithm adds $w^{2^1}, w^{2^2}, \ldots, w^k$ (we are interested only in $w^k$); finally, the strings $w_k[2..k]w^k, w_k[3..k]w^k, \ldots, w_k[k..k]w^k$ are added. So, the algorithm adds to the dictionary exactly the following strings:
\begin{itemize}
\item all prefixes and suffixes of $w_i$ (including $w_i$ itself) for $i = 1,2,\ldots,k$;
\item $w_{k-2}w_{k-1}, w_{k-3}w_{k-2}w_{k-1}, \ldots, w_1w_2\cdots w_{k-1}$, and $w$;
\item $w^k$ along with $w^{k/2},\ldots, w^{2^2}, w^2$ (we use only $w^k$ in what follows);
\item $w_k[2..k]w^k, w_k[3..k]w^k, \ldots, w_k[k..k]w^k$.
\end{itemize}

It is easy to verify that $|w| = k^2$, $|w^k| = k^3$, and $|s'| = \Theta(k^4)$. (The string $w_k[2..k]w^k w_k[3..k]\-w^k \cdots w_k[k..k]w^k$ contributes the most to the length.)

We first provide an overview of our construction. The main load on the running time of the algorithm is concentrated in the following strings $z_i$:
$$
z_i = w_i[2..k]w_{i+1}\cdots w_k w^{k-2}w_1\cdots w_i\text{ for }i = 1,2,\ldots,k - 2.
$$
Put $s = s'x_1z_1\SEP\SEP x_2z_2\SEP\SEP \cdots x_{k-2}z_{k-2}\SEP\SEP$, where $x_1, \ldots, x_k$ are auxiliary strings defined below. Before processing of $z_i$, the algorithm processes $x_i$ and adds the strings $w_i[j..k]w_{i+1}\cdots w_{k-1}\-w_k[1..j{-}1]$ and $w_k[j..k]w_1\cdots w_{i-1}w_i[1..j]$ for $j \in [2..k]$ to the dictionary (see below). So, analyzing $z_i$, the algorithm consecutively ``jumps'', for $j = 2,3,\ldots, k$, from the string $w_i[j..k]w_{i+1}\cdots w_{k-1}w_k[1..j{-}1]$ to $w_k[j..k]w_1\cdots w_{i-1}w_i[1..j]$ and so on. The crucial point is that, while analyzing $w_k[j..k]w_1\cdots w_{i-1}w_i[1..j]$, the algorithm does not know in advance that the string $w_k[j..k]w^k$ from the dictionary does not occur at this position and, since the length of the longest common prefix of the strings $w_k[j..k]w^k$ and $w_k[j..k]w^{k-j}w_1\cdots w_i\SEP\SEP$ is $\Theta(k-j+1 + |w^{k-j}|)$, spends $\Theta(|w^{k-j}|) = \Theta((k-j)k^2)$ time verifying this. Therefore, the analysis of the string $s$ takes $\Theta((k - 2)\sum_{j=2}^{k} (k-j)k^2) = \Theta(k^5)$ time overall. Since $|z_i| = O(k^3)$ and, as it is shown below, $|x_i| = O(k^3)$, we have $n = |s| = \Theta(k^4)$ and the processing time is $\Theta(n^{\frac{5}4})$ as required. We now describe this in more detail.

We prove by induction that the following invariant is maintained: when the algorithm starts the processing of the suffix $x_iz_i\SEP\SEP\cdots x_{k-2}z_{k-2}\SEP\SEP$ of the string $s$ ($x_i$ are defined below), the dictionary contains the following set of strings:
\begin{itemize}
\item ``building blocks'' constructed during the processing of $s'$;
\item pairs of separators $\SEP\SEP$ (recall that all separators are distinct);
\item for each $i' \in [1..i{-}1]$ and $j \in [2..k]$:
\subitem $w_{i'}[j..k]w_{i'+1}\cdots w_{k-1}w_k[1..j{-}1]$ and $w_k[j..k]w_1\cdots w_{i'-1}w_{i'}[1..j]$,
\subitem $w_{i'}[j..k]w_{i'+1}\cdots w_{k-1}$ and $w_k[j..k]w_1\cdots w_{i'-1}$,
\subitem $w_{i'}[j..k]w_{i'+1}\cdots w_k w_1\cdots w_{i'-1}w_{i'}[1..j]$.
\end{itemize}

The strings from the last two lines in the above list are not used and appear as byproducts. (But it is still important to have them in mind to verify that the algorithm works as expected.)  So, assume that, by inductive hypothesis, the invariant holds for all $i' \in [1..i{-}1]$ (it is trivial for $i = 1$).

Define $x_i$ as follows (the parentheses are only for visual ease):
$$
\begin{array}{l}
u'_{i,j} = (w_k[j..k]w_1\cdots w_{i-1} w_i[1..j]),\\
u_{i,j} = (w_k[j..k]w_1\cdots w_{i-2}w_{i-1}[1..j]) (w_{i-1}[j{+}1..k]) u'_{i,j},\\
v_{i,j} = (w_i[j..k]w_{i+1}\cdots w_{k-1}) (w_i[j..k]w_{i+1}\cdots w_{k-1}w_k[1..j{-}1]),\\
x_1 = (u'_{1,2}\SEP\SEP u'_{1,3}\SEP\SEP\cdots u'_{1,k-1}\SEP\SEP u'_{1,k}\SEP\SEP) (v_{1,2}\SEP\SEP v_{1,3}\SEP\SEP\cdots v_{1,k}\SEP\SEP),\\
x_i = (u_{i,2}\SEP\SEP u_{i,3}\SEP\SEP\cdots u_{i,k-1}\SEP\SEP u'_{i,k}\SEP\SEP) (v_{i,2}\SEP\SEP v_{i,3}\SEP\SEP\cdots v_{i,k}\SEP\SEP),\text{ for }i\ne 1.
\end{array}
$$

Observe that $|x_i| = O(k^3)$. Using the inductive hypothesis, one can prove that the algorithm adds the strings $w_k[j..k]w_1\cdots w_{i-1}$ (for $j \ne k$), $w_k[j..k]w_1\cdots w_{i-1}w_i[1..j]$, $w_i[j..k]w_{i+1}\cdots w_{k-1}$, and $w_i[j..k]w_{i+1}\cdots w_{k-1}w_k[1..j{-}1]$ for $j \in [2..k]$ to the dictionary after the processing of $x_i$ (plus several pairs $\SEP\SEP$). It remains to show that the algorithm adds exactly the strings $w_i[j..k]w_{i+1}\cdots w_k w_1 \cdots w_{i-1}w_i[1..j]$, for $j \in [2..k]$, to the dictionary when processing $z_i$.

Observe that, for $j \in [2..k]$, $w_i[j..k]w_{i+1}\cdots w_{k-1}w_k[1..j{-}1]$ is the longest string from the dictionary that has prefix $w_i[j..k]$, and $w_k[j..k]w_1\cdots w_{i-1}w_i[1..j]$ is the longest string from the dictionary that has prefix $w_k[j..k]$ and does not coincide with $w_k[j..k]w^k$. Hence, the algorithm consecutively ``jumps'' over the substrings $w$ of the string $z_i$ adding after each such ``jump'' the string $w_i[j..k]w_{i+1}\cdots w_k w_1 \cdots w_{i-1}w_i[1..j]$ to the dictionary (for $j = 2,3,\ldots,k$). No other strings are added.

Each time the algorithm processes a substring $w_k[j..k]w_1\cdots w_{i-1}w_i[1..j]$, it also verifies in $\Theta(ki + |w^{k-j}|)$ time whether the string $w_k[j..k]w^k$ occurs at this position. Therefore, by the above analysis, processing takes $\Theta(|s|^{\frac{5}4})$ time.

An analogous troublesome string for the na\"{\i}ve LZMW construction algorithm is as follows (again, all separators ${\SEP}$ are assumed to be distinct letters):
$$
\begin{array}{l}
w_i = a_{i,1}a_{i,2} \cdots a_{i,k}\text{ and }w = w_1w_2 \cdots w_k,\\
s_{pre,i} = w_i[1..2]\SEP w_i[1..3]\SEP \cdots\SEP w_i[1..k]\SEP,\\
s_{suf,i} = w_i[k{-}1..k]\SEP w_i[k{-}2..k]\SEP \cdots\SEP w_i[2..k]\SEP,\\
p = s_{pre,1} s_{pre,2} \cdots s_{pre,k} s_{suf,1} s_{suf,2} \cdots s_{suf,k},\\
q = w_{k-2}w_{k-1}\SEP w_{k-3}w_{k-2}w_{k-1}\SEP \cdots\SEP w_1w_2{\cdots}w_{k-1}\SEP w\SEP,\\
s' = p q w^{2^1}\SEP w^{2^2}\SEP \cdots\SEP w^k\SEP w_k[2..k]w^k\SEP w_k[3..k]w^k\SEP \cdots\SEP w_k[k..k]w^k\SEP,\\
y_j = w_k[j..k]w_1\SEP w_k[j..k]w_1w_2[1..j]\SEP,\\
t_{i,j} = w_{i-2}[j{+}1..k]w_{i-1}[1..j]\SEP w_{i-1}[j{+}1..k]w_i[1..j],\\
u_{i,j} = (w_k[j..k]w_1\cdots w_{i-3}w_{i-2}[1..j]) (w_{i-2}[j{+}1..k]w_{i-1}[1..j]),\\
v_{i,j} = w_i[j..k]w_{i+1}\cdots w_{k-1}\SEP w_i[j..k]w_{i+1}\cdots w_{k-1}w_k[1..j{-}1],\\
x_i = t_{i,2}\SEP t_{i,3}\SEP\cdots\SEP t_{i,k-1}\SEP u_{i,2}\SEP u_{i,3}\SEP\cdots\SEP u_{i,k}\SEP v_{i,2}\SEP v_{i,3}\SEP\cdots\SEP v_{i,k}\SEP,\\
z_i = w_i[2..k]w_{i+1}\cdots w_k w^{k-2}w_1\cdots w_i\SEP,\\
s = s'y_2 y_3\cdots y_k x_4 z_4 x_6 z_6 \cdots x_{2j} z_{2j} \cdots  x_{k-2} z_{k-2}.
\end{array}
$$

Let us explain on a high level why the LZMW algorithm works slowly on $s$.
While analyzing the prefix $s' y_2 y_3 \cdots y_k$, the algorithm adds a number of ``building block'' strings into the LZMW dictionary, including the strings $w[j..k]w^k$ for $j = 2,3,\ldots,k$ (recall that, unlike the LZD dictionary containing phrases, the LZMW dictionary contains pairs of adjacent phrases).
Before the processing of $z_i$, the algorithm processes $x_i$ and adds the strings $w_i[j..k]w_{i+1}\cdots w_{k-1}w_k[1..j{-}1]$ (from $v_{i,j}$), $w_k[j..k]w_1\cdots w_{i-2}w_{i-1}[1..j]$ (from $u_{i,j}$), and $w_{i-1}[j{+}1..k]w_i[1..j]$ (from $t_{i,j}$) to the dictionary. The concatenation of these three strings is $w_i[j..k]w_{i+1}\cdots w_kw_1\cdots w_{i-1}w_i[1..j]$, so, analyzing $z_i$, the algorithm consecutively ``jumps'', for $j = 2,3,\ldots, k$, from the string $w_i[j..k]w_{i+1}\cdots w_{k-1}w_k[1..j{-}1]$ to $w_k[j..k]w_1\cdots w_{i-2}w_{i-1}[1..j]$ and then to $w_{i-1}[j{+}1..k]w_i[1..j]$, thus producing three new phrases (and then moves on to $j{+}1$). The point is that, while analyzing the string $w_k[j..k]w_1\cdots w_{i-2}w_{i-1}[1..j]$, the algorithm does not know in advance that the string $w_k[j..k]w^k$ from the dictionary does not occur at this position and, since the length of the longest common prefix of the strings $w_k[j..k]w^k$ and $w_k[j..k]w^{k-j}w_1\cdots w_i\SEP\SEP$ is $\Theta(k-j+1 + |w^{k-j}|)$, spends $\Theta(|w^{k-j}|) = \Theta((k-j)k^2)$ time verifying this. Therefore, the analysis of the string $s$ takes $\Theta((k/2)\sum_{j=2}^{k} (k-j)k^2) = \Theta(k^5)$ time overall. Since $n = |s| = \Theta(k^4)$, the processing time is $\Theta(n^{\frac{5}4})$ as required. We omit the detailed proof since it is very similar to the LZD case.

To additionally verify the correctness of both constructed examples, we performed the na\"{\i}ve LZD and LZMW algorithms (with some diagnostics to track their execution) on the examples for small $k$ and, indeed, observed the expected ``bad'' behavior in the special positions described above. Our verifying code (it can be found in~\cite{LZDandLZMWcode}) thoroughly checks the correspondence of the behavior of the parsers in the special positions to the behavior discussed in the above text. Thus, we hope that the correctness of both our constructions is well supported.
\cqed
\end{proof}

We now explain how to decrease the alphabet size in the examples of  Theorem~\ref{LowerboundThm}.
The construction for both parsing schemes relies on the following reduction.
\begin{lemma}\label{lem:reduction}
Consider the parsing scheme  LZD or LZMW and a string $s\in \Sigma^*$.
There exists a string $t\in \{0,1\}^*$ of length $\Theta(|\Sigma|\log |\Sigma|)$
and a morphism $\phi$ with $\phi(\Sigma)\subseteq \{0,1\}^{\ell}$ for $\ell = \Theta(\log |\Sigma|)$ such that the parsing of $t \cdot \phi(s)$ consists
of the parsing of $t$ followed by the image with respect to $\phi$ of the parsing of $s$.
\end{lemma}
\begin{proof}
We analyze the two parsing schemes separately.
For LZD, we recursively define $A_L \subseteq \{0,1\}^{2^L}$, setting $A_0 = \{0,1\}$
and $A_L = \{xy : x,y\in A_{L-1} \wedge x \le y\}$ for $L>0$.
Let $(\alpha_i)_{i=1}^\infty$ be the infinite sequence of all elements of $A_L$, for all $L\ge 1$,
with members of each set $A_L$ listed in the lexicographic order; e.g., $\alpha_1,\ldots, \alpha_{12} = 00, 01, 11, 0000, 0001, 0011, 0101, 0111, 1111, 00000000, 00000001, 00000011$.
We will define $t=\alpha_1\cdots \alpha_m$ for some $m$.  Let us characterize parsings of such strings.
\begin{claim}
For any non-negative integer $m$ and any string $w\in \{0,1\}^*$, the first $m$ phrases of the LZD parsing of the binary string $\alpha_1\cdots \alpha_m\cdot w$ are $\alpha_1,\ldots,\alpha_m$.
\end{claim}
\begin{proof}
We proceed by induction on $m$; the base case of $m=0$ is trivial.

For $m>0$, the inductive assumption implies that the first $m-1$  phrases are $\alpha_1,\ldots,\alpha_{m-1}$.
Our goal is to prove that the $m$th phrase is $\alpha_m$.
Before processing $\alpha_m$, the LZD dictionary is $D = \{0, 1, \alpha_1,\ldots,\alpha_{m-1}\}$.
Suppose that $\alpha_m=xy \in A_L$ with $x,y \in A_{L-1}$. Recall that $x\le y$; consequently, $D\cap \left(y\cdot \{0,1\}^*\right) = \{y\}$ and
$$ D\cap \left(x\cdot \{0,1\}^*\right) = \{x\} \cup \{xy' : y'\in A_{L-1} \wedge x \le y' < y\}.$$
Thus, the longest prefix of $\alpha_m\cdot w$ contained in $D$ is $x$, and the longest prefix of $y\cdot w$ contained in $D$ is $y$.
This means that the $m$th phrase is indeed $\alpha_m=xy$.\cqed
\end{proof}

Consider a string $s\in \Sigma^n$.
We choose the smallest $L$ with $|A_L|\ge |\Sigma|$ and define $t=\alpha_1\cdots \alpha_m$ so that $t$ is shortest possible
and the LZD dictionary after processing $t$ contains at least $|\Sigma|$ elements of $A_L$.
The morphism $\phi$ is then defined by injectively mapping $\Sigma$ to these dictionary strings from $A_L$.

Note that $|A_{L-1}| \le |\Sigma|$ and $m \le |\Sigma| + \sum_{\ell=1}^{L-1} |A_\ell|$, so we have $m = \Theta(|\Sigma|)$,
$\ell = 2^L=\Theta(\log |\Sigma|)$, and $|t| = \Theta(|\Sigma|\log |\Sigma|)$, as desired.

We are to prove that the LZD parsing of $t\cdot\phi(s)$ is $\alpha_1,\ldots,\alpha_m,\phi(p_1),\ldots,\phi(p_z)$,
where $p_1,\ldots,p_z$ is the LZD parsing of $s$. For this, we inductively prove that the LZD dictionary $D$ after parsing $p_1\cdots p_i$ is related
to the LZD dictionary $\hat{D}$ after parsing $t\cdot \phi(p_1\cdots p_i)$ by the following invariant: $\hat{D} \cap \left(\phi(\Sigma)\cdot \{0,1\}^*\right) = \phi(D).$
The base case follows from the claim ($\hat{D}\cap \left(\phi(\Sigma)\cdot \{0,1\}^*\right) = \phi(\Sigma)=\phi(D)$), and the inductive step is straightforward.
This completes the proof for the LZD scheme.

The construction for LZMW is more involved, but the idea is the same.
We recursively define $B_L \subseteq \{0,1\}^{2^L}$, setting $B_0 = \{0,1\}$
and $B_L = \{xy : x,y\in B_{L-1} \wedge xy \ne 1^{2^{L-1}}0^{2^{L-1}}\}$ for $L>0$.
Let $(\beta_i)_{i=1}^\infty$ be the infinite sequence that lists all elements of $B_L$ consecutively for all $L\ge 0$,
with members of each $B_L$ listed in the lexicographic order (i.e., $(\beta_i)_{i=1}^\infty$ is defined by analogy with $(\alpha_i)_{i=1}^\infty$ for LZD but starting with $L = 0$).
For $\beta_m \in B_L$, define $b(\beta_m) = \beta_M \beta_m\cdot \beta_{M+1}\beta_m \cdots\beta_{m-1} \beta_{m}\cdot \beta_{m}$,
where $\beta_M=0^{2^{L}}$ is the first element of $B_L$ in $(\beta_i)_{i=1}^\infty$.
For example, $b(\beta_1)\cdots b(\beta_6) = 0\cdot 0 \ 1\ 1\cdot  00\cdot 00\ 01\ 01\cdot  00\ 11\ 01\ 11\ 11\cdot 0000.$
\begin{claim}
For $m\ge 1$, consider a binary string $ b(\beta_1)\cdots b(\beta_m)\cdot 0^{|\beta_m|} \cdot w$ for $w\in \{0,1\}^*$.
The LZMW parsing decomposes its fragments $b(\beta_i)$ into phrases of length $|\beta_i|$.
\end{claim}
\begin{proof}
We proceed by induction on $m$. The base case $m=1$ is straightforward: it suffices to note that the first phrase of $0\cdot 0 \cdot w$ is $0$.
Below, we consider $m>1$.

First, suppose that $\beta_m = 0^{2^L}$, i.e., $\beta_{m-1} = 1^{2^{L-1}} \in B_{L-1}$.
Note that $b(\beta_m)$ starts with $0^{2^{L-1}}$, so the inductive hypothesis
yields that the prefix $ b(\beta_1)\cdots b(\beta_{m-1})$ is parsed as desired.
Observe that after parsing this prefix, the LZMW dictionary
is $D=\{1^{2^{\ell-1}} 0^{2^{\ell}} : 0 < \ell < L\} \cup \bigcup_{\ell=0}^{L} B_{\ell}.$
Consequently, we obtain $D\cap \left(B_L \cdot \{0,1\}^*\right) = B_L$ and, therefore, $b(\beta_m)=\beta_m$ is parsed as claimed.

Finally, suppose that $\beta_m\in B_L \setminus \{0^{2^L}\}$. In this case, $\beta_{m-1}\in B_L$ and $\beta_M= 0^{2^L}$ for some $M < m$.
Since $b(\beta_m)$ starts with $\beta_M = 0^{2^L}$, the inductive hypothesis lets us assume that the prefix $ b(\beta_1)\cdots b(\beta_{m-1})$ is parsed as desired.
Due to $1^{2^{L-1}}0^{2^L-1}\notin B_L$, after parsing this prefix, the LZMW dictionary $D$ satisfies:
$$D\cap (B_L\cdot \{0,1\}^{*}) = B_L \cup \{\beta_k\beta_{k'} : M \le k,k' < m\wedge  (k,k')\ne (m{-}1,M)\}.$$
Let us consider the parsing of $b(\beta_m)0^{2^L} w = \beta_M \beta_m \cdot \beta_{M+1}\beta_m\cdots \beta_{m-1}\beta_{m}\cdot \beta_{m}\cdot 0^{2^L} w$. One can inductively prove that
before parsing $\beta_k \beta_m \cdot \beta_{k+1}\cdots$, for $M\le k < m$, we have $D\cap \left(\beta_k \cdot \{0,1\}^*\right) = \{\beta_k\}\cup \{\beta_k \beta_{k'} : M \le k' < m\}$,
so the subsequent phrase is $\beta_k$. Next, before parsing $\beta_m\cdot \beta_{k+1} \cdots$, for $M\le k < m$, we have
$D\cap \left(\beta_m\cdot \{0,1\}^*\right) = \{\beta_m\}\cup \{\beta_m \beta_{k'} : M < k' \le k\}$, so the subsequent phrase is $\beta_m$.
Finally, before parsing $\beta_m\cdot 0^{2^L} w$, we have
$D\cap \left(\beta_m\cdot \{0,1\}^*\right) = \{\beta_m\}\cup \{\beta_m \beta_{k'} : M < k' < m\}$, so the last phrase is also $\beta_m$.
Thus, $b(\beta_m)$ is parsed as claimed.\cqed
\end{proof}
Consider a string $s\in \Sigma^n$.
We choose the smallest $L$ with $|B_L|\ge |\Sigma|$ and define $t=b(\beta_1)\cdots b(\beta_m)$ so that $t$ is shortest possible
and the LZMW dictionary after processing $t$ contains at least $|\Sigma|$ members of $B_L$ (note that $\beta_m \in B_{L-1}$ in this case).
The morphism $\phi$ is then defined by injectively mapping $\Sigma$ to these dictionary strings from $B_L$.
Moreover, we put $\phi(s[1]) = 0^{2^L}$ so that the claim is applicable for $t\cdot \phi(s)$.
The remaining proof is analogous to the LZD counterpart. We only need to observe that the LZMW dictionary
additionally contains  $\beta_m0^{2^L}$, but $\beta_m 0^{2^{L-1}}\notin \phi(\Sigma)$ and, hence, this does not affect the parsing of $t\cdot\phi(s)$.  \cqed
\end{proof}

The hard binary examples are now straightforward to derive.
\begin{theorem}
The na\"{\i}ve LZD and LZMW parsing algorithms take time $\Omega(n^{5/4} / \log^{1/4} n)$ in the worst case even on a binary alphabet.
\end{theorem}
\begin{proof}
We apply Lemma~\ref{lem:reduction} for a string $s\in \Sigma^*$ of length $n$ constructed in the proof of Theorem~\ref{LowerboundThm}
for the appropriate parsing algorithm, which results in a binary string $\hat{s}:=t\cdot \phi(s)$.
Without loss of generality, we may assume $|\Sigma|\le n$, so  $\hat{n} := |\hat{s}|= \Theta(|\Sigma|\log |\Sigma| + n \log |\Sigma|) = \Theta(n \log |\Sigma|).$
Recall that the na\"{\i}ve parsing algorithm traverses at least $\Omega(n^{5/4})$ trie edges while parsing $s$.
Since the parsing of the suffix $\phi(s)$ of $\hat{s}$ is the $\phi$-image of the parsing of $s$,
this algorithm traverses at least $\Omega(n^{5/4}\log |\Sigma|)$ trie edges while parsing $\hat{s}$.
In terms of $\hat{n}$, the running time is at least $\Omega(\hat{n}^{5/4} / \log^{1/4} |\Sigma|)$, which is $\Omega(\hat{n}^{5/4} / \log^{1/4} \hat{n})$ due to $|\Sigma| \le n < \hat{n}$.
\cqed
\end{proof}

\section{Faster Small-Space Computation}
\label{sec-faster}

In this section we describe a new parsing algorithm that works in $O(n + z\log^2 n)$ time (randomized, in expectation)
and uses $O(z\log n)$ working space to parse the input string over the integer alphabet $\{0,1,\ldots,n^{O(1)}\}$. Before getting to the algorithm itself, we review four tools that are essential for
it: Karp--Rabin hashing~\cite{KR87}, AVL-grammars of Rytter~\cite{Rytter03}, the dynamic z-fast trie of Belazzougui et al.~\cite{BBV10},
and the dynamic marked ancestor data structure of Westbrook~\cite{Westbrook92}.

\paragraph{Karp--Rabin hashing.}
A Karp--Rabin~\cite{KR87} hash function $\phi$ has the form
$\phi (s [1..n]) = \sum_{i = 1}^n s [i] \delta^{i - 1} \bmod p$, where
$p$ is a fixed prime and $\delta$ is a randomly chosen integer from the range $[0..p{-}1]$ (this is a more popular version of the original hash proposed in~\cite{KR87}).  The value $\phi(s)$ is called $s$'s Karp--Rabin hash. It is well-known that, for any $c > 3$, if $p > n^c$, then the probability that two distinct substrings of the given input string of length $n$ have the same hash is less than~$\frac{1}{n^{c-3}}$.

We extensively use the property that the hash of the concatenation $s_1s_2$ of two strings $s_1$ and $s_2$ can be computed as $(\phi(s_1) + \delta^{|s_1|}\phi(s_2)) \bmod p$. Therefore, if the values $\phi(s_1)$ and $\phi(s_2)$ are known and $p \le n^{O(1)}$, then $\phi(s_1s_2)$ can be calculated in $O(1)$ time provided the number $(\delta^{|s_1|} \bmod p)$ is known.

\paragraph{AVL-grammars.}

Consider a context-free grammar $G$ that generates a string $s$ (and only $s$). Denote by $Tree(G)$ the derivation tree of $s$.
We say that $G$ is an \emph{AVL-grammar} (see~\cite{Rytter03}) if $G$ is in the Chomsky normal form and, for every internal node $v$ of $Tree(G)$, the heights of the trees rooted at the left and right children of $v$ differ by at most $1$.
The following result straightforwardly follows from the algorithm of Rytter described in~\cite{Rytter03}.
\begin{lemma}[{see \cite[Th. 2]{Rytter03}}]
Let $G$ be an AVL-grammar generating a prefix $s[1..i{-}1]$ of a string $s$. Suppose that the string $s[i..k]$ occurs in $s[1..i{-}1]$; then one can construct an AVL-grammar generating the string $s[1..k]$ in $O(\log i)$ time modifying at most $O(\log i)$ rules in $G$.\label{lem:addPhrase}
\end{lemma}

Let $G$ be an AVL-grammar generating a string $s$. It is well-known that, for any substring $s[i..j]$, one can find in $O(\log n)$ time $O(\log n)$ non-terminals $A_1, \ldots, A_k$ such that $s[i..j]$ is equal to the string generated by $A_1\cdots A_k$. Hence, if each non-terminal $A$ of $G$ is augmented with the Karp--Rabin hash $\phi(t)$ of the string $t$ generated by $A$ and with the number $\delta^{|t|} \bmod p$, then we can compute $\phi(s[i..j])$ in $O(\log n)$ time. One can show that, during the reconstruction of the AVL-grammar in Lemma~\ref{lem:addPhrase}, it is easy to maintain the described integers augmenting the non-terminals (see~\cite{Rytter03}).

\paragraph{Z-fast tries.}

Let $x$ be a string such that one can compute the Karp--Rabin hash of any prefix of $x$ in $O(t_x)$ time.
The z-fast trie~\cite{BBV10} is a compacted trie containing a dynamic set of variable-length strings that supports the following operations:
\begin{itemize}
\item we can find (w.h.p.) in $O(t_x\log |x|)$ time the highest explicit node $v$ such that the longest prefix of $x$ present in the trie is written on the root-$v$ path;
\item we can insert $x$ into the trie in $O(|x| + t_x\log |x|)$ randomized time.
\end{itemize}
The space occupied by the z-fast trie is $\Theta(k)$, where $k$ is the number of strings inserted in the trie.
%In fact a wider range of operations and with slightly better
%time bounds than what we state above are described by Belazzougui et al.~\cite{BBPV09} --- we state slower bounds
%here for clarity of exposition and because they are sufficient to achieve our claims.
%The key to the lookup operation is a technique called {\em fat-binary search}. We refer the interested reader
%to~\cite{BBV10} for the details.

\paragraph{Dynamic marked ancestor.}

Let $T$ be a dynamic compacted trie (or just tree) with $k$ nodes. The dynamic marked ancestor data structure of~\cite{Westbrook92} supports the following two operations on $T$ (both in $O(\log k)$ time): for a given node $v$, (1)~mark $v$, (2)~find the nearest marked ancestor of $v$ (if any).

\paragraph{Algorithm.}

Our faster parsing algorithm computes the LZD phrases from left to right one by one, spending $O(\log^{O(1)} n)$ time on each phrase. We maintain an AVL-grammar $G$ for the prefix $s[1..i{-}1]$ of $s$ we have already parsed, and a z-fast trie $T$ containing the first phrases $p_1, p_2, \ldots, p_r$ of the LZD parsing of $s$ such that $s[1..i{-}1] = p_1p_2\cdots p_r$. We augment $T$ with the dynamic marked ancestor data structure and mark all nodes corresponding to phrases (i.e., all nodes $v$ such that the string written on the path from the root to $v$ is equal to $t \in \{p_1,\ldots,p_r\}$). We augment each non-terminal of $G$ with the Karp--Rabin hash $\phi(t)$ of this non-terminal's expansion $t$ and with the number $\delta^{|t|} \bmod p$, so that the hash of any substring of $s[1..i{-}1]$ can by calculated in $O(\log n)$ time.
%Since, as it is shown below, we spend only $O(\log^{O(1)} n)$ time on each phrase, the time required to compute the first $\sqrt{n}$ phrases is upper bounded by $O(\sqrt{n}\log^{O(1)}n) = O(n)$ and, hence, for simplicity of the exposition, we can assume hereafter that the aforementioned table of size $O(\sqrt{n})$ is already precomputed.

%and we use Westbrook's~\cite{Westbrook92} dynamic marked-ancestor data structure to keep track of the nodes in $T$ whose path labels are current phrases (although the complete path labels are not stored in $T$ and, if we want them, we must use $G$ as well as $T$ to recover them) and we store the starting points of the phrases in those nodes.

Suppose we are looking for the first part of the next phrase and that, in addition to having parsed $s[1..i{-}1]$, we have already read $s[i..j{-}1]$ without parsing it~--- but we have found the endpoints of an occurrence of $s[i..j{-}1]$ in $s[1..i{-}1]$.  (Notice $s[i..j{-}1]$ can be empty, i.e., $i = j$.) Denote by $x$ the longest prefix of $s[i..j{-}1]$ that is also a prefix of some of the phrases $p_1,\ldots,p_r$. Since we can  compute quickly with $G$ the hash of any prefix of $s[i..j{-}1]$, we can use the z-fast search to find in $O(\log^2 n)$ time a node $v$ of $T$ such that $x$ is written on the path connecting the root and $v$. Let $s[\ell_v .. r_v]$ be a substring of $s[1..i{-}1]$ corresponding to $v$ (the numbers $\ell_v$ and $r_v$ are stored with the node $v$). Using hashes and the binary search, we find the longest common prefix of the strings $s[i..j{-}1]$ and $s[\ell_v .. r_v]$ (with high probability) in $O(\log^2 n)$ time; this prefix must be $x$.
%use doubling search and searches in $T$ to find the longest prefix of \(s [i..j{-}1]\) that is also a prefix of a current phrase, in $O (\log^2 n)$ time, verifying $T$'s answers with $G$ (i.e., turning weak prefix search into high-probability prefix search).

If $s[i..j{-}1] \ne x$, then we perform a marked-ancestor query on the vertex corresponding to $x$ (which can be found in $O(\log^2 n)$ time in the same way as $v$) and thus find the longest phrase that is a prefix of $s[i..j{-}1]$.  We take that phrase as the first part of the next phrase and start over, looking for the second part, with the remainder of \(s [i..j{-}1]\) now being what we have read but not parsed (of which we know an occurrence in $s [1..i{-}1]$).  On the other hand, if $s[i..j{-}1] = x$, then we read $s[j..n]$ in blocks of length $\log^2 n$, stopping when we encounter an index $k$ such that $s[i..k]$ is not a prefix of a phrase $p_1,\ldots,p_r$; the details follow.

Suppose that we have read $q$ blocks and the concatenation $s[i..j + q \log^2 n - 1]$ of $s[i..j{-}1]$ and the $q$ previous blocks is a prefix of a phrase $t \in \{p_1, \ldots, p_r\}$. We compute in $O(\log^2 n)$ time the hashes of all the prefixes of the block $s[j + q \log^2 n .. j + (q + 1) \log^2 n - 1]$, which allows us to compute the hash of any prefix of $s[i .. j + (q + 1)\log^2 n - 1]$ in $O(\log n)$ time.
%We can find such a $k$ quickly by computing the hashes of all the prefixes of each block $s[j + q \log^2 n .. j + (q + 1) \log^2 n - 1]$ we read, in $O(\log^2 n)$ time. Since the concatenation $s[i..j + q \log^2 n - 1]$ of $s[i..j{-}1]$ and the $q$ previous blocks we have read is a prefix of a current phrase, we can compute the hash of any prefix of that concatenation in $O (\log n)$ time with $G$.
Therefore, again using z-fast search and binary search, we can check in $O(\log^2 n)$ time if the block $s[j + q \log^2 n .. j + (q + 1) \log^2 n - 1]$ contains such a $k$ --- and, if so, find it.  If $k$ is not found, then using information from the search, we can find a phrase \(t' \in \{p_1, \ldots, p_r\}\) --- which may or may not be equal to $t$ --- such that $s[i .. j + (q + 1)\log^2 n - 1]$ is a prefix of $t'$; we then proceed to the $(q{+}2)$nd block.

Once we have found such a $k$, we conceptually undo reading the characters from $s[k]$ onwards (which causes us to re-read later those $O(\log^2 n)$ characters), then perform a search and marked-ancestor query in $T$, which returns the longest phrase that is a prefix of $s[i..k{-}1]$.  We take that longest phrase as the first part of the next phrase and start over, looking for the second part, with the remainder of \(s [i..k{-}1]\) now being what we have read but not parsed (of which we know an occurrence in $s[1..i{-}1]$).

Once we have found both the first and second parts of the next phrase~--- say, $p'_1$ and $p'_2$~--- we add the next phrase $p_{r+1} = p'_1p'_2$ to $G$ (by Lemma~\ref{lem:addPhrase}) and to $T$, which takes $O(|p_{r+1}| + \log^2 n)$ time.  In total, since processing each block takes $O(\log^2 n)$ time and the algorithm processes at most $z + \frac{n}{\log^2 n}$ blocks, we parse $s$ in $O(n + z \log^2 n)$ time.  Our space usage is dominated by $G$, which takes $O(z \log n)$ space. Finally, we verify in a straightforward manner in $O(n)$ time whether the constructed parsing indeed encodes the input string. If not (which can happen with probability $\frac{1}{n^{c-3}}$, where $p > n^c$), we choose a different random $\delta \in [0..p{-}1]$ for the Karp--Rabin hash and execute the whole algorithm again.

The computation of the LZMW parsing in $O(n + z\log^2 n)$ expected time and $O(z\log n)$ space is similar: the z-fast trie stores pairs $p_1p_2, p_2p_3, \ldots, p_{z-1}p_z$ of adjacent phrases in this case and the nodes corresponding to these pairs are marked. We omit the details as they are straightforward.

% -- if this is true, then it should be at least an appendix -- Dima
%We can derandomize our algorithm at the cost of increasing the time and space bounds by using, e.g., Mehlhorn, Sundar and Uhrig's~\cite{MSU97} dynamic equality-testing data structure instead of Karp--Rabin hashing.

% -- why do we promise this again, while we are not even sure if we plan to do it? -- Dima
%Of course, we will give more details in the full version of this paper.

\section{Concluding Remarks}
\label{sec-conclusion}

We believe that our new parsing algorithms can be implemented efficiently, and we leave this as future work. Perhaps a more interesting question is whether there exists an LZD/LZMW parsing algorithm with better working space and the same (or better) runtime.  We note that the algorithmic techniques we have developed here can also be applied to, e.g., develop more space-efficient parsing algorithms for LZ-End~\cite{KN13}, a variant of LZ77~\cite{ZL77} with which each phrase \(s [i..j]\) is the longest prefix of \(s [i..n]\) such that an occurrence of \(s [i..j{-}1]\) in \(s [1..i{-}1]\) ends at a phrase boundary.  Kempa and Kosolobov~\cite{KK17} very recently gave an LZ-End parsing algorithm that runs in $O (n \log \ell)$ expected time and $O (z + \ell)$ space, where $\ell$ is the length of the longest phrase and $z$ is the number of phrases.

To reduce Kempa and Kosolobov's space bound, we keep an AVL-grammar (again augmented with the non-terminals' Karp--Rabin hashes, meaning our algorithm Monte-Carlo) of the prefix of $s$ we have processed so far; a list of the endpoints of the phrases so far, in the right-to-left lexicographic order of the prefixes ending at the phrases' endpoints; and an undo stack of the phrases so far.  For each character \(s [k]\) in turn, for \(1 \leq k \leq n\), in $O(\log^{O (1)} n)$ time we use the grammar and the list to find the longest suffix \(s [j..k]\) of \(s [1..k]\) such that an occurrence of \(s [j..k{-}1]\) in \(s [1..j{-}1]\) ends at a phrase boundary.  We use the undo stack to remove from the grammar, the list, and the stack itself, all the complete phrases lying in the substring $s [j..k{-}1]$, and then add the phrase consisting of the concatenation of those removed phrases and \(s [k]\).  By~\cite[Lemma~3]{KK17}, we remove at most two phrases while processing \(s [k]\), so we still use a total of $O(\log^{O (1)} n)$ worst-case time for each character of $s$.  Again, the space bound is dominated by the grammar, which takes \(O (z \log n)\) words.  We leave the details for the full version of this paper.

% -- If this is true, then this must be our result (check it, please):
% -- We don't have much to gain from including this. Leaving it commented out for now - could be good to explore it in a journal version.
%We note that our algorithm, probably, can be derandomized by using the deterministic string hashing of Melhorn et al.~\cite{MSU97} in place of Karp--Rabin hashing.

Regarding compression performance, we have shown that like their ancestor, LZ78, both LZD and LZMW sometimes approximate the smallest grammar poorly. This, of course, does not necessarily detract from their usefulness in real compression tools; now however, practitioners have a much clearer picture of these algorithms' possible behavior. The future work includes closing the gap between the lower bound $\Omega(n^{\frac{1}{3}})$ and the upper bound $O((n/\log n)^{\frac{2}{3}})$ for the approximation ratio and designing parsing algorithms with better guarantees.

\section*{Acknowledgements}
We thank H. Bannai, P. Cording, K. Dabrowski, D. Hücke, D. Kempa, L. Salmela for interesting discussions on LZD at the 2016 StringMasters and Dagstuhl meetings. Thanks also go to D. Belazzougui for advice about the z-fast trie and to the anonymous referees.

\bibliographystyle{plainurl}
\bibliography{refs}

\begin{thebibliography}{10}

\bibitem{LZDandLZMWcode}
Supplementary materials for the present paper: C++ code for described
  experiments.
\newblock \url{https://bitbucket.org/dkosolobov/lzd-lzmw}.

\bibitem{BBV10}
Djamal Belazzougui, Paolo Boldi, and Sebastiano Vigna.
\newblock Dynamic z-fast tries.
\newblock In {\em Proc. 17th International Symposium on String Processing and
  Information Retrieval (SPIRE)}, volume 6393 of {\em LNCS}, pages 159--172.
  Springer, 2010.
\newblock \href {http://dx.doi.org/10.1007/978-3-642-16321-0_15}
  {\path{doi:10.1007/978-3-642-16321-0_15}}.

\bibitem{BCPT15}
Djamal Belazzougui, Patrick~Hagge Cording, Simon~J. Puglisi, and Yasuo Tabei.
\newblock Access, rank, and select in grammar-compressed strings.
\newblock In {\em Proc. 23rd Annual European Symposium on Algorithms (ESA)},
  volume 9294 of {\em LNCS}, pages 142--154. Springer, 2015.
\newblock \href {http://dx.doi.org/10.1007/978-3-662-48350-3_13}
  {\path{doi:10.1007/978-3-662-48350-3_13}}.

\bibitem{BLRSSW15}
Philip Bille, Gad~M. Landau, Rajeev Raman, Kunihiko Sadakane, Srinivasa~Rao
  Satti, and Oren Weimann.
\newblock Random access to grammar-compressed strings and trees.
\newblock {\em {SIAM} Journal on Computing}, 44(3):513--539, 2015.
\newblock \href {http://dx.doi.org/10.1137/130936889}
  {\path{doi:10.1137/130936889}}.

\bibitem{CharikarEtAl}
Moses Charikar, Eric Lehman, Ding Liu, Rina Panigrahy, Manoj Prabhakaran, Amit
  Sahai, and abhi shelat.
\newblock The smallest grammar problem.
\newblock {\em IEEE Transactions on Information Theory}, 51(7):2554--2576,
  2005.
\newblock \href {http://dx.doi.org/10.1109/TIT.2005.850116}
  {\path{doi:10.1109/TIT.2005.850116}}.

\bibitem{CN11}
Francisco Claude and Gonzalo Navarro.
\newblock Self-indexed grammar-based compression.
\newblock {\em Fundamenta Informaticae}, 111(3):313--337, 2011.
\newblock \href {http://dx.doi.org/10.3233/FI-2011-565}
  {\path{doi:10.3233/FI-2011-565}}.

\bibitem{GGKNP12}
Travis Gagie, Paweł Gawrychowski, Juha K{\"{a}}rkk{\"{a}}inen, Yakov Nekrich,
  and Simon~J. Puglisi.
\newblock A faster grammar-based self-index.
\newblock In {\em Proc. 6th International Conference on Language and Automata
  Theory and Applications ({LATA})}, volume 7183 of {\em LNCS}, pages 240--251.
  Springer, 2012.
\newblock \href {http://dx.doi.org/10.1007/978-3-642-28332-1_21}
  {\path{doi:10.1007/978-3-642-28332-1_21}}.

\bibitem{GotoEtAl}
Keisuke Goto, Hideo Bannai, Shunsuke Inenaga, and Masayuki Takeda.
\newblock {LZD} factorization: simple and practical online grammar compression
  with variable-to-fixed encoding.
\newblock In {\em Proc. 25th Annual Symposium on Combinatorial Pattern Matching
  (CPM)}, volume 9133 of {\em LNCS}, pages 219--230. Springer, 2015.
\newblock \href {http://dx.doi.org/10.1007/978-3-319-19929-0_19}
  {\path{doi:10.1007/978-3-319-19929-0_19}}.

\bibitem{HuckeLohreyReh}
Danny Hücke, Markus Lohrey, and Carl~Philipp Reh.
\newblock The smallest grammar problem revisited.
\newblock In {\em Proc. 23rd International Symposium on String Processing and
  Information Retrieval (SPIRE)}, volume 9954 of {\em LNCS}, pages 35--49.
  Springer, 2016.
\newblock \href {http://dx.doi.org/10.1007/978-3-319-46049-9_4}
  {\path{doi:10.1007/978-3-319-46049-9_4}}.

\bibitem{INIBT13}
Tomohiro I, Yuto Nakashima, Shunsuke Inenaga, Hideo Bannai, and Masayuki
  Takeda.
\newblock Efficient {L}yndon factorization of grammar compressed text.
\newblock In {\em Proc. 24th Annual Symposium on Combinatorial Pattern Matching
  ({CPM})}, volume 7922 of {\em LNCS}, pages 153--164. Springer, 2013.
\newblock \href {http://dx.doi.org/10.1007/978-3-642-38905-4_16}
  {\path{doi:10.1007/978-3-642-38905-4_16}}.

\bibitem{KR87}
Richard~M. Karp and Michael~O. Rabin.
\newblock Efficient randomized pattern-matching algorithms.
\newblock {\em IBM Journal of Reseach and Development}, 31(2):249--260, 1987.
\newblock \href {http://dx.doi.org/10.1147/rd.312.0249}
  {\path{doi:10.1147/rd.312.0249}}.

\bibitem{KK17}
Dominik Kempa and Dmitry Kosolobov.
\newblock {LZ-End} parsing in compressed space.
\newblock In {\em Proc.\ Data Compression Conference (DCC)}, pages 350--359.
  IEEE, 2017.
\newblock \href {http://dx.doi.org/10.1109/DCC.2017.73}
  {\path{doi:10.1109/DCC.2017.73}}.

\bibitem{KN13}
Sebastian Kreft and Gonzalo Navarro.
\newblock On compressing and indexing repetitive sequences.
\newblock {\em Theoretical Computer Science}, 483:115--133, 2013.
\newblock \href {http://dx.doi.org/10.1016/j.tcs.2012.02.006}
  {\path{doi:10.1016/j.tcs.2012.02.006}}.

\bibitem{MillerWegman}
Victor~S. Miller and Mark~N. Wegman.
\newblock Variations on a theme by {Z}iv and {L}empel.
\newblock In {\em Proc. NATO Advanced Research Workshop on Combinatorial
  Algorithms on Words}, volume~12 of {\em NATO ASI}, pages 131--140. Springer,
  1985.
\newblock \href {http://dx.doi.org/10.1007/978-3-642-82456-2_9}
  {\path{doi:10.1007/978-3-642-82456-2_9}}.

\bibitem{Rytter03}
Wojciech Rytter.
\newblock Application of {Lempel-Ziv} factorization to the approximation of
  grammar-based compression.
\newblock {\em Theoretical Computer Science}, 302(1-3):211--222, 2003.
\newblock \href {http://dx.doi.org/10.1016/S0304-3975(02)00777-6}
  {\path{doi:10.1016/S0304-3975(02)00777-6}}.

\bibitem{TIIBT13}
Toshiya Tanaka, Tomohiro I, Shunsuke Inenaga, Hideo Bannai, and Masayuki
  Takeda.
\newblock Computing convolution on grammar-compressed text.
\newblock In {\em Proc. Data Compression Conference ({DCC})}, pages 451--460.
  {IEEE}, 2013.
\newblock \href {http://dx.doi.org/10.1109/DCC.2013.53}
  {\path{doi:10.1109/DCC.2013.53}}.

\bibitem{Westbrook92}
Jeffery Westbrook.
\newblock Fast incremental planarity testing.
\newblock In {\em Proc. 19th International Colloquium on Automata, Languages
  and Programming (ICALP)}, volume 623 of {\em LNCS}, pages 342--353. Springer,
  1992.
\newblock \href {http://dx.doi.org/10.1007/3-540-55719-9_86}
  {\path{doi:10.1007/3-540-55719-9_86}}.

\bibitem{ZL77}
Jacob Ziv and Abraham Lempel.
\newblock A universal algorithm for sequential data compression.
\newblock {\em {IEEE} Transactions on Information Theory}, 23(3):337--343,
  1977.
\newblock \href {http://dx.doi.org/10.1109/TIT.1977.1055714}
  {\path{doi:10.1109/TIT.1977.1055714}}.

\bibitem{ZL78}
Jacob Ziv and Abraham Lempel.
\newblock Compression of individual sequences via variable-rate coding.
\newblock {\em {IEEE} Transactions on Information Theory}, 24(5):530--536,
  1978.
\newblock \href {http://dx.doi.org/10.1109/TIT.1978.1055934}
  {\path{doi:10.1109/TIT.1978.1055934}}.

\end{thebibliography}
\end{document}